\documentclass[aps,prl,showpacs,twocolumn,amsmath,amsfonts,amssymb,superscriptaddress,tightenlines]{revtex4-1}

\usepackage[colorlinks,citecolor=blue,linkcolor=cyan]{hyperref}
\usepackage{grffile}
\usepackage[caption=false]{subfig}
\usepackage{graphicx}
\usepackage{mathtools}

\usepackage{paralist}

\newcommand{\dg}{^{\dagger}}
\newcommand{\ndg}{^{\phantom{}}}
\renewcommand{\vec}[1]{{\boldsymbol{#1}}}
\newcommand{\op}{\hat}

\newcommand{\wt}{\widetilde}

\newcommand{\MaxMin}{\textsf{MaxMin}}

\newcommand{\Sgn}{\textsf{Sgn}}

\newcommand{\smat}[1]{\left(
    \begin{smallmatrix}
      #1
    \end{smallmatrix}
  \right)}

\begin{document}
\date{\today}

\title{Dynamical quantum phase transitions: \\Role of topological
  nodes in wavefunction overlaps}

\author{Zhoushen Huang} \affiliation{Institute for Materials Science,
  Los Alamos National Laboratory, Los Alamos, NM 87545, USA}
\email{zsh@lanl.gov}
\author{Alexander V.~Balatsky} \affiliation{Institute for Materials
  Science, Los Alamos National Laboratory, Los Alamos, NM 87545, USA}
\affiliation{NORDITA, Roslagstullsbacken 23, SE-106 91\ \ Stockholm,
  Sweden}
\email{avb@nordita.org}

\begin{abstract}
  A sudden quantum quench of a Bloch band from one topological phase
  toward another has been shown to exhibit an intimate connection with
  the notion of a dynamical quantum phase transition (DQPT), where the
  returning probability of the quenched state to the initial
  state---i.e. the Loschmidt echo---vanishes at critical times
  $\{t^{*}\}$. Analytical results so far are limited to two-band
  models, leaving the exact relation between topology and DQPT
  unclear. In this work, we show that for a general multi-band system,
  a robust DQPT relies on the existence of nodes (i.e. zeros) in the
  wavefunction overlap between the initial band and the post-quench
  energy eigenstates. These nodes are topologically protected if the
  two participating wavefunctions have distinctive topological
  indices. We demonstrate these ideas in detail for both one and two
  spatial dimensions using a three-band generalized Hofstadter
  model. We also discuss possible experimental observations.
\end{abstract}
\maketitle

\paragraph{Introduction} Advances in experimental techniques, in
particular in cold atom systems \cite{Bloch08, Bloch12, Langen15},
have reinvigorated recent interest in quantum dynamics
\cite{Polkovnikov11}.  A paradigmatic setup in this context is a
quantum quench \cite{Igloi00, Altman02, Polkovnikov02, Sengupta04,
  Calabrese06}, wherein a system is prepared as an eigenstate
$|\Psi\rangle$ of an initial Hamiltonian $H_I$, but evolved under a
different Hamiltonian $H_F$. In a slow ramp \cite{Barankov08,
  Pollmann10}, one has in addition the control over how fast the
switching between $H_I$ and $H_F$ can be, as well as what path to take
in the space of Hamiltonians. Since $|\Psi\rangle$ typically consists
of many excited states of $H_F$ with a non-thermal distribution, its
time evolution provides a unique venue for investigating issues in
nonequilibrium quantum statistical mechanics such as thermalization,
equilibration, or the lack thereof \cite{Polkovnikov11, Nandkishore15,
  Eisert15, DAlessio15, Essler16, Millen16}. A particularly fruitful
approach to understanding dynamics after a quantum quench is by
exploiting the formal similarity between the time evolution operator
$\exp(-iHt)$, and the thermal density operator $\exp(-\beta H)$. This
enables one to leverage and extend notions in equilibrium statistical
mechanics to the realm of quantum dynamics. In this spirit, the return
amplitude
\begin{gather}
  \label{gt-def}
  G(t) = \langle \Psi | e^{-iH_F t} | \Psi\rangle = \sum_n
  \left|\langle \Phi^{(n)} | \Psi \rangle\right|^2 e^{-iE_n t}
\end{gather}
can be thought of as a partition function along imaginary temperature
$\beta = it$, with the prepared state $|\Psi\rangle$ as a fixed
boundary \cite{LeClair95}. Here $|\Phi^{(n)}\rangle$ and $E_n$ are
eigenstates and eigenvalues of the post-quench $H_F$,
respectively. Heyl \emph{et al} showed \cite{Heyl13} that analogous to
the thermal free energy, a dynamical free energy density
\cite{Fagotti13} can be defined, $f(t) = -\log G(t)/L$, where $L$ is
system size. Singularities in $f$ then signifies the onset of what was
proposed as a \emph{dynamical quantum phase transition} (DQPT). In
statistical mechanics, phase transitions are closely related to the
zeros of the partition function---known as Fisher zeros---in the
complex temperature plane \cite{Fisher65}. Historically, Yang and Lee
were the first to connect phase transitions with zeros of the
partition function in complexified parameter space \cite{Yang52,
  *Lee52}.  While Fisher zeros are always complex for finite systems,
they may coalesce into a continuum (line in one parameter dimension,
area in two parameter dimensions, etc) that cuts through the real
temperature axis in the thermodynamic limit, giving rise to an
equilibrium phase transition. Investigations on DQPT have followed a
similar route by first solving the Fisher zeros in the complex
temperature plane, and then identifying conditions for them to cross
the axis of imaginary temperature (real time). DQPT is thus
mathematically identified as $G(t^{*})=0$ at critical time(s) $t^{*}$
\footnote{The formal similarity between $e^{-iHt}$ and $e^{-\beta H}$
  is routinely used in field theoretic calculations such as
  correlators, where the thermal and quantum-dynamic results are
  related through Wick rotation. Knowledge in one domain can then be
  transcribed into the other without additional calculation. This
  ``equivalence'', however, relies on the existence of an analytic
  continuation between $it$ and $\beta$ in the complex temperature
  plane. The presence of closed and/or semi-infinite Fisher zero lines
  divides the complex temperature plane into disjoint regions, and the
  correspondence between $it$ and $\beta$ will break down if they
  reside in different regions. For quenches that exhibit DQPT, as
  pointed out in Ref.~\cite{Heyl13}, this implies that non-equilibrium
  time evolution (i.e. along the imaginary temperature axis) is in
  general not governed by the equilibrium thermodynamics (along the
  real temperature axis).}. DQPTs occur in both integrable
\cite{Heyl13, Pozsgay13, Vajna14, Hickey14, Vajna15, Schmitt15,
  Divakaran16, Sharma16, Budich16} and non-integrable
\cite{Karrasch13, Fagotti13, Kriel14, Andraschko14, Sharma15} spin
systems for quenches across quantum critical points. They can further
be classified by discontinuities in different orders of time
derivatives of $f(t)$ \cite{Canovi14, Schmitt15} \emph{vis-a-vis}
their thermal counterparts.  Very recently DQPTs have also been shown
to constitute \emph{unstable} fixed points in the renormalization
group flow, and are therefore subject to the notion of universality
class and scaling \cite{Heyl15}.

Physically, the return amplitude $G(t)$ is related to the power
spectrum of work performed during a quench,
$G(\omega) = \sum_n \left| \langle \Phi^{(n)} | \Psi\rangle \right|^2
\delta\left(\omega - (E_n - E_I)\right)$, which is the Fourier
transform of $G(t)e^{-iE_It}$, and $E_I$ is the energy of the initial
state \cite{Talkner07, Campisi11, Jarzynski97, Deffner15}. This in
principle makes $G(t)$, and hence DQPT, a measurable phenomenon. A
pratically more viable route to experimental verification is through
measuring time evolution of thermodynamic quantities, which may
exhibit post-quench oscillations at a time scale commensurate with the
DQPT critical time $t^{*}$, and universal scaling near $t^{*}$
\cite{Heyl14}. In band systems, as we will show, they may also be
identified by a complete depletion at $t^{*}$ of sublattice or
spin-polarized particle density at certain crystal momenta, see
Eq.~\ref{rho-kn}.

Parallel to the development of DQPT as the dynamical analogue of
equilibrium phase transitions is the investigation on its relation
with topology \cite{Vajna14, Hickey14, Vajna15, Schmitt15}. This issue
arises naturally because in the transverse field Ising model, in which
DQPT was first discovered, the quantum critical point can be mapped to
a topological phase transition at which the quantized Berry phase of
the \emph{fermionized} Hamiltonian jumps between $0$ and $\pi$. DQPT
in this two-band fermion model was attributed to the occurrence of
``population inversion'' \cite{Heyl13} where it becomes equally
probable to find the initial state in either of the two post-quench
bands, a consequence of the Berry phase jump \cite{Vajna15}. The same
analysis has been extended to various two-band models in one- and
two-spatial dimensions (1D/2D) \cite{Schmitt15, Vajna15,Divakaran16,
  Hickey14, Sharma16}, where definitive connection was found between
DQPT and quench across topological transitions, although some
complications exist \footnote{E.g., in 2D integer quantum Hall
  systems, DQPT is related to the change in the \emph{absolute value}
  of Chern number $|C|$ instead of $C$ \cite{Vajna15}}. DQPTs have also
been demonstrated to occur for quenches within the same topological
phase \cite{Vajna14, Andraschko14, Hickey14}, although from the point
of view of topological protection, these are not robust as they
require fine-tuning of the Hamiltonians.

The purpose of this work is to develop a general theory \emph{beyond
  two band models} to clarify the relation between \emph{robust} DQPT
and topology. We will show that a robust DQPT---one which is
insensitive to the details of the pre- and post-quench Hamiltonians
other than the phases to which they belong---relies on the existence
of zeros (or nodes) in the wavefunction overlap between the initial
band and all eigenstates of the post-quench Hamiltonian. These nodes
are topologically protected if the two participating wavefunctions
have distinctive topological indices: for example, the Chern number
difference $|C_{\psi} - C_{\phi}|$ provides a lower bound to the
number of $\vec k$-space nodes in the overlap
$\langle \phi_{\vec k}|\psi_{\vec k}\rangle$, see
Theorem~\ref{thm-2d}.  These considerations lead to the notion of
\emph{topological} and \emph{symmetry-protected DQPTs} which we will
demonstrate in detail using a $3$-band generalized Hofstadter model.
Analysis of a 1D 3-band model exhibiting symmetry-protected DQPT can
be found in Supplemental Materials (SM).
\paragraph{Amplitude and phase conditions of DQPT}
The DQPT condition $G(t^{*}) = 0$ can be interpreted geometrically as
the complex numbers
$z_n(t) = \left|\langle \Phi^{(n)} | \Psi\rangle\right|^2 e^{-iE_n t}$
forming a closed polygon in the complex plane at $t^{*}$, see
Fig.~\ref{fig:z-polygon}.  The time-independent content of this
observation is that the amplitudes $\{\left|z_n\right|\}$ satisfy a
generalized triangle inequality,
$\sum_{m\neq n} |z_m| \ge |z_n| \ \forall n$. Invoking
$\langle \Psi | \Psi\rangle = \sum_n |z_n| = 1$, one has the
\emph{amplitude condition},
\begin{gather}
  \label{tri-ineq}
  |z_n| = \left| \langle \Phi^{(n)} | \Psi\rangle \right|^2
  \stackrel{!}{\le} \frac{1}{2} \quad \forall n\ .
\end{gather}
For $\{|z_n|\}$ that satisfy Eq.~\ref{tri-ineq}, solutions to
$\sum_n |z_n| e^{-i\varphi_n} = 0$ exist and form a subspace
$\mathcal{M}_{\{|z_n|\}}$ on the $N$-torus,
\begin{gather}
  \label{Mz}
  \mathcal{M}_{\{|z_n|\}}\in \mathcal{T}^N: \left\{
    \{e^{-i\varphi_n}\} \bigg| \sum_{n=1}^N |z_n|e^{-i\varphi_n} = 0
  \right\}\ .
\end{gather}
To set off DQPT, the dynamical phases must be able to evolve into
$\mathcal{M}_{\{|z_n|\}}$. This constitutes the \emph{phase
  condition},
\begin{gather}
  \label{phase-cond}
  \exists t^{*} : \{e^{-iE_n t^{*}}\} \in \mathcal{M}_{\{|z_n|\}} \ .
\end{gather}
DQPT requires both conditions to hold simultaneously.

\begin{figure}
  \centering
  \includegraphics[width=.34\textwidth,trim={20 20 20 20},clip]{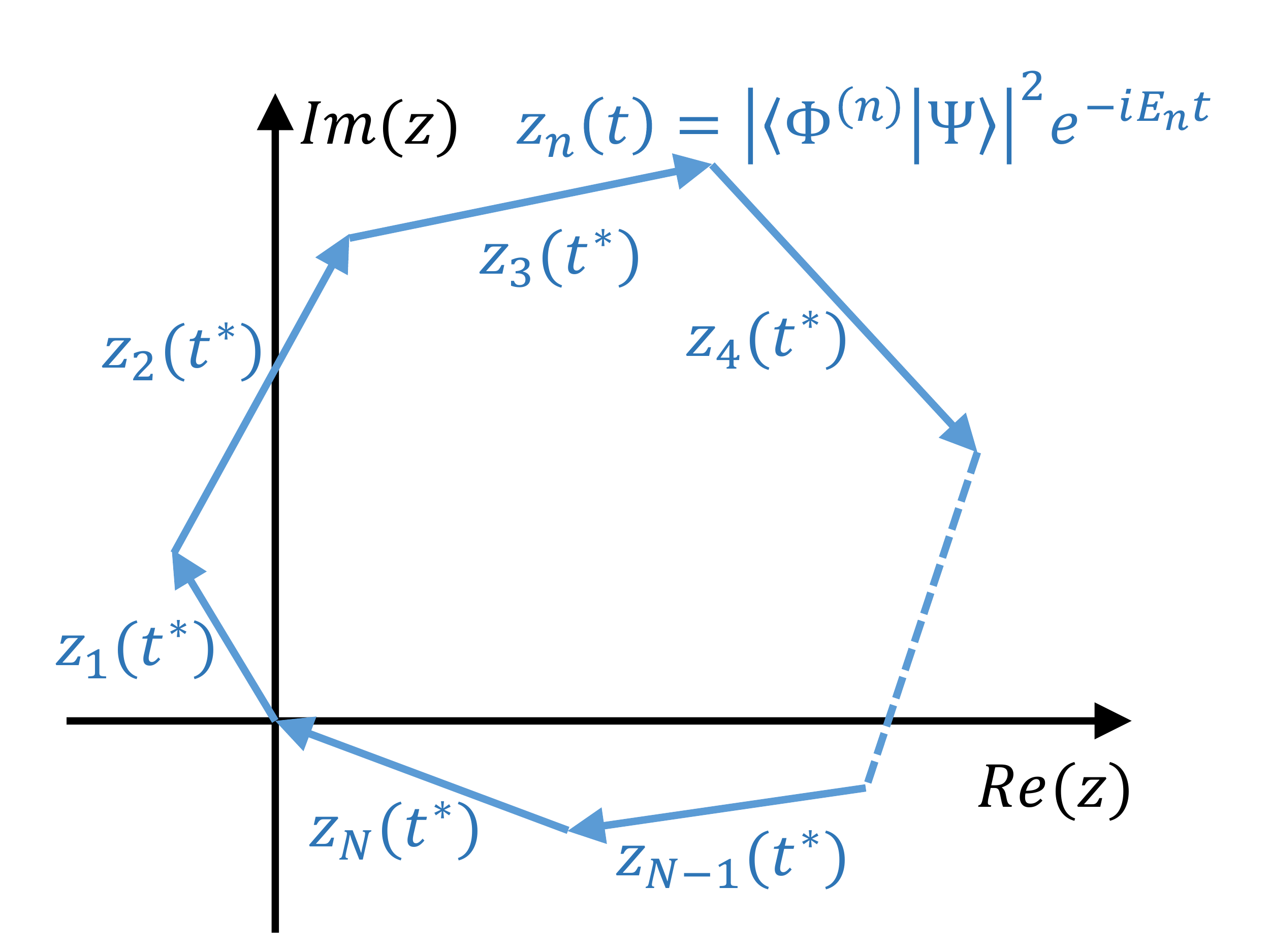}
  \caption{(Color online) Geometric representation of the DQPT
    condition $G(t^{*}) = \sum_n z_n(t^{*})=0$. $\{z_n(t^{*})\}$
    must form a closed polygon in the complex plane, and hence satisfy
    a generalized triangle inequality $|z_n| \le \sum_{m\neq n}|z_m|$.
    Wavefunction normalization
    $\langle \Psi | \Psi\rangle = \sum_n |z_n|=1$ then leads to
    $|z_n| \le \frac{1}{2}$ }
  \label{fig:z-polygon}
\end{figure}

\paragraph{Phase ergodicity in few-level systems} At first glance, the
phase condition may seem to be the more stringent one. After a quench
across a quantum phase transition, a many-body initial state
$|\Psi\rangle$ typically has overlap with an extensive amount of
eigenstates of the post-quench Hamiltonian $H_F$ and therefore the
amplitudes $\langle \Phi^{(n)} | \Psi\rangle$ are generically
exponentially small in system size, rendering Eq.~\ref{tri-ineq}
satisfied in general. Existence of DQPT then relies entirely on the
phase condition.  Integrable systems, however, point to the
possibility that the amplitude and phase conditions may be intricately
related and traded for one another. Such systems can effectively be
broken down into few-level subsystems labeled by quantum numbers
$\vec k$, say $N_{\vec k}$ levels $\{E_{\vec k,n}\}$ for
$n = 1,2, \cdots, N_{\vec k}$ in the $\vec k$ sector. Correspondingly
$G(t) = \prod_{\vec k}G(\vec k,t)$. For the transverse field Ising
model, Kitaev's honeycomb model \cite{Kitaev06}, and band insulator
models, $\vec k$ is the Bloch momentum. It is known that as long as
the $N_{\vec k} - 1$ gaps,
$\Delta_{\vec k,n} = E_{\vec k,n+1} - E_{\vec k,n}$, are not
rationally related, the dynamical phases $\{e^{-iE_{\vec k,n}t}\}$ are
\emph{ergodic} on the $N_{\vec k}$-torus up to an overall phase
\cite{Cornfeld12}, and \emph{will} therefore evolve into its subspace
$\mathcal{M}_{\{|z_n|\}}$ (Eq.~\ref{Mz}). Phase ergodicity thus
guarantees the phase condition Eq.~\ref{phase-cond}, and DQPT in each
$\vec k$ sector depends entirely on the amplitude condition.

\paragraph{Robust DQPT protected by nodes in wavefunction overlap}
Hereafter, we focus on quenches in multi-band Bloch systems with $N_B$
bands.  For simplicity we use a single filled band
$|\psi(\vec k)\rangle$ as the pre-quench state. Generalization to
multiple filled bands is straightforward.  The post-quench return
amplitude is $G(t) = \prod_{\vec k}G(\vec k, t)$,
\begin{gather}
  G(\vec k, t) = \sum_{n=1}^{N_B} \left| \langle \phi^{(n)}(\vec k) |
    \psi(\vec k)\rangle \right|^2 e^{-i\varepsilon_n(\vec k) t}\ ,
\end{gather}
where $|\phi^{(n)}(\vec k)\rangle$ and $\varepsilon_n(\vec k)$ are
respectively the post-quench energy eigenstates and
eigenvalues. \emph{Assume phase ergodicity holds at all $\vec k$
  points}---this is a very relaxed requirement provided there is no
degeneracy at any $\vec k$ point.  Then DQPT amounts to the existence
of at least one $\vec k$ at which Eq.~\ref{tri-ineq} is satisfied,
namely
\begin{gather}
  \label{tri-ineq-k}
  \exists \vec k \in \textsf{Brillouin Zone}: \left|\langle
    \phi^{(n)}(\vec k) | \psi(\vec k)\rangle\right|^2 \le \frac{1}{2}\,
  \forall n \ .
\end{gather}

We now discuss how Eq.~\ref{tri-ineq-k} and hence DQPT can arise from
nodes in wavefunction overlaps. Note that this is \emph{not} the only
way to get DQPT. Its virtue lies in its robustness against
perturbations to the Hamiltonians. In SM, we provide examples where
DQPTs with no overlap node can be easily avoided simply by Hamiltonian
parameter tuning without crossing a phase boundary. The overlap nodes
are, on the other hand, typically topologically protected, a point we
will return to later.  Now consider the following quench. Let
$a = 1,2, \cdots, N_B$ label ``sublattices'', which in general may
also include other degrees of freedom, e.g., orbitals, spins, etc.
Prepare the pre-quench state by filling $a=1$,
\begin{gather}
  |\Psi\rangle = \prod_{\vec r} \psi_{\vec r,1}\dg | \emptyset\rangle =
  \prod_{\vec k} \psi_{\vec k,1}\dg |\emptyset\rangle\ ,
\end{gather}
where $\psi_{\vec r,a}\dg$ creates an electron on sublattice $a$ in
unit cell $\vec r$, $|\emptyset\rangle$ is the vacuum,
$\psi_{\vec k,a}\dg = \frac{1}{\sqrt{N}}\sum_{\vec r} e^{i\vec k\cdot
  \vec r} \psi_{\vec r,a}\dg$, and $N$ is the total number of unit
cells. The system is then time-evolved under an integer quantum Hall
Hamiltonian $\op H = \sum_{\vec k} \op H(\vec k)$ where
$\op H(\vec k) = \sum_{a,b = 1}^{N_B} H_{a,b}\ndg(\vec k) \psi_{\vec
  k, a}\dg \psi_{\vec k,b}\ndg = \sum_{n=1}^{N_B} \varepsilon_{\vec
  k,n}\ndg\phi_{\vec k, n}\dg \phi_{\vec k, n}\ndg$, and we assume the
Chern number of all bands of $\op H(\vec k)$ are non-zero,
$C_n \neq 0\, \forall n$. The overlap in Eq.~\ref{tri-ineq-k} is
$\langle \emptyset | \phi_{\vec k, n}\ndg \psi_{\vec k, 1}\dg
|\emptyset\rangle = \phi^{(n)}_1(\vec k)^{*}$, where
$\phi^{(n)}_a(\vec k) = \langle a | \phi^{(n)}(\vec k)\rangle$ is the
$a^{th}$ component of
$|\phi^{(n)}(\vec k)\rangle = (\phi^{(n)}_1(\vec k), \phi^{(n)}_2(\vec
k), \cdots, \phi^{(n)}_{N_B}(\vec k))^t$, an eigenvector of the
post-quench Hamiltonian matrix $H(\vec k)$. It is known that any
component $\phi_a^{(n)}(\vec k) \,\forall a$ must have \emph{at least}
$|C_n|$ zeros in the Brillouin zone \cite{Kohmoto85}, see also
Thm.~\ref{thm-2d}. Now assume at an arbitrary Bloch momentum
$\vec k_0$, $\phi^{(n_1)}_1$ has the highest weight:
$|\phi^{(n_1)}_1(\vec k_0)| > |\phi^{(n\neq n_1)}_1(\vec k_0)|$. The
existence of node means $\phi^{(n_1)}_{1}$ cannot remain as the
highest weight element over the entire Brillouin zone, and hence must
switch rank with the second highest weight element, say
$\phi^{(n_2)}_1$, at some point $\vec k_c$:
$|\phi^{(n_1)}_1(\vec k_c)| = |\phi^{(n_2)}_1(\vec k_c)| \ge
|\phi^{(n\neq n_1, n_2)}_1(\vec k_c)|$ \footnote{for two band models,
  $\vec k_c$ is where the so-called ``population inversion'' occurs.}.
Together with the normalization
$\langle \emptyset | \psi_{\vec k, 1}\dg \psi_{\vec k,
  1}\ndg|\emptyset\rangle = \sum_n |\phi^{(n)}_1|^2 = 1$, one
concludes that at $\vec k = \vec k_c$, Eq.~\ref{tri-ineq-k} is
satisfied.

Note that in this case, the return amplitude $G(\vec k, t)$ is related
to the $\vec k$-space sublattice particle density,
\begin{gather}
  \label{rho-kn}
  \rho_{\vec k,a}(t) \equiv \langle \Psi(t) | \psi_{\vec k, a}\dg \psi_{\vec
    k,a}\ndg | \Psi(t)\rangle = \left|G(\vec k,t)\right|^2\ .
\end{gather}
A DQPT can thus be identified by $\rho_{\vec k, a}(t^{*}) = 0$, i.e.,
a complete depletion of particles with momentum $\vec k$ on sublattice
$a$ (or orbital, spin, etc.), which may be experimentally measurable.

The argument above for node-protected DQPT applies to any
pre-/post-quench combinations. In general, if the overlap of the
pre-quench band $|\psi(\vec k)\rangle$ with \emph{every} eigenstate
$|\phi^{(n)}(\vec k)\rangle$ of $H_F(\vec k)$ has nodes in the
Brillouin zone, then the triangle inequality Eq.~\ref{tri-ineq-k} is
guaranteed, and a robust DQPT would occur. This criterion can be
written in a form more amenable to numerical test,
\begin{gather}
  \label{maxmin1}
  \psi_{\MaxMin} \equiv \max_n\Bigl[\min_{\vec k} \bigl| \langle
  \phi^{(n)}(\vec k) | \psi(\vec k)\rangle \bigr|\Bigr]\ , \\
  \label{maxmin2}
  \psi_{\MaxMin} = 0 \Leftrightarrow \text{Robust DQPT}\ .
\end{gather}

\paragraph{Topological protection of nodes in wavefunction overlaps}
There is a curious connection between wavefunction zeros and
quantization.  In elementary quantum mechanics, nodes in the radial
wavefunction is related to the principal quantum number
\cite{LandauQM}. In continuum integer quantum Hall systems, the number
of nodes in the wavefunction
$\psi(\vec r) = \langle \vec r | \psi\rangle$ for $\vec r$ in a
magnetic unit cell is given by its Chern number magnitude $|C|$
\cite{Kohmoto85}. These nodes persist even in the presence of weak
disorder \cite{Arovas88}.  On a lattice, $|C|$ gives the number of
$\vec k$-space nodes in all wavefunction components
$\psi_a(\vec k) = \langle a | \psi(\vec k)\rangle \forall a$
\cite{Kohmoto85}, a phenomenon closely related to the energetic
spectral flow of the edge states \cite{Hatsugai93}. Note that the
relation between $C$ and wavefunction nodes relies on one participant
of the overlap, namely the basis states $|\vec r\rangle$ and
$|a\rangle$, to be topologically trivial. If both participants can be
nontrivial, the number of nodes in their overlap should depend on both
topological indices on an equal footing. Indeed we have the following
theorems, \newtheorem{theorem}{Theorem}
\begin{theorem}
  \label{thm-2d} In 2D, the overlap of Bloch bands
  $|\psi(\vec k)\rangle$ and $|\phi(\vec k)\rangle$, with Chern
  numbers $C_{\psi}$ and $C_{\phi}$ respectively, must have at least
  $|C_{\psi} - C_{\phi}|$ nodes in the Brillouin zone.
\end{theorem}
\begin{theorem}
  \label{thm-1d} 
  In 1D, the Berry phase $\gamma$ of a real Bloch band,
  $|\psi(k)\rangle = (\psi_1(k), \psi_2(k), \cdots)^t, \psi_a(k) \in
  \mathbb{R} \,\forall a$, is quantized to $0$ or $\pi$. The overlap
  of two real bands $|\psi(k)\rangle$ and $|\phi(k)\rangle$, with
  Berry phases $\gamma_{\psi}$ and $\gamma_{\phi}$ respectively, must
  have at least one node if $\gamma_{\psi} \neq \gamma_{\phi}$.
\end{theorem}
See SM for proof. Note that symmetry protection may enforce a
Hamiltonian to be real \cite{Fang15}, leading to the real bands in
Thm.~\ref{thm-1d}. This prompts the notion of \emph{symmetry-protected
  DQPT}, reminiscent of symmetry-protected topological phases that may
be classified by topological numbers at high-symmetry hyper-surfaces
\cite{Ryu10, Alexandradinata14, Fang15}. An example will be given
later, see also SM.

\paragraph{Generalized Hofstadter model}
We demonstrate ideas discussed above using a generalized Hofstadter
model,
\begin{gather}
  \label{hof-hk}
  H(\vec k, t, m) =
  \begin{pmatrix}
    d_1 & v_1 & v_3e^{ik_y}\\
    v_1 & d_2 & v_2 \\
    v_3 e^{-ik_y} & v_2 & d_3
  \end{pmatrix} \ , \\
  \notag
  d_a = 2\cos(k_x + a \varphi) + am\ , \\
  \notag
  v_a = 1+2t\cos\left[k_x + (a +\frac{1}{2})\varphi\right]\ , \\
  \notag
  a = 1,2,3\ , \ \varphi=\frac{2\pi}{3}\ .
\end{gather}
The nearest neighbor hopping is set as $1$. At $t = m = 0$, we recover
the Hofstadter model \cite{Hofstadter76, Hatsugai93, Zhao11, Huang12,
  Wang13, Harper14, Aidelsburger15} on a square lattice with magnetic
flux $\varphi$ per structural unit cell, and its magnetic unit cell
consists of $3$ structural unit cells along the $y$ direction.
$t \neq 0$ allows for second neighbor (i.e.~diagonal) hopping, and
$m\neq 0$ describes a flux-commensurate onsite sawtooth potential. See
SM for phase diagram.  At $k_y = 0$ and $\pi$, $H(\vec k)$ is
invariant under the combined transformation of time-reversal,
$H(\vec k) \rightarrow H^{*}(-\vec k)$, and inversion,
$H(\vec k) \rightarrow H(-\vec k)$, and is hence real. Eigenstates
there are subject to Thm.~\ref{thm-1d}.

Now consider quenches in which the initial state is prepared by
filling one of the three bands of a pre-quench Hamiltonian
parameterized by $t_i, m_i$, and evolved using a post-quench
Hamiltonian with $t_f, m_f$ \footnote{an initial state of two filled
  bands is equivalent to one with a single filled band through
  particle-hole transformation.}. In Fig.~\ref{fig:maxmin}, we keep
$t_i, m_i, m_f$ fixed, and plot the $\MaxMin$ (Eq.~\ref{maxmin1}) of
the three pre-quench bands as functions of the post-quench $t_f$. By
varying $t_f$, the post-quench $H(\vec k)$ is swept through six
different topological phases as labeled in Fig.~\ref{fig:maxmin}. 

Let us illustrate topological and symmetry-protected DQPTs with two
examples, using $\psi^{(2)}$ as the pre-quench state (blue circled
line in Fig.~\ref{fig:maxmin}):
\begin{inparaenum}[(i)]
\item \label{topo-dqpt}\emph{Topological DQPT} protected by 2D Chern
  number. Consider the quench from $\psi^{(2)}$ to phase 5. In this case,
  the Chern number of the pre-quench state ($C=-1$) differs from
  \emph{all} three Chern numbers of the post-quench Hamiltonian
  ($C=[1,-2,1]$), thus from Thm.~\ref{thm-2d}, all three overlaps have
  nodes, and Eq.~\ref{tri-ineq-k} is satisfied.
\item \label{sp-dqpt}\emph{Symmetry-protected DQPT}. Consider the
  quench from $\psi^{(2)}$ to phase 2. In this case, the pre-quench
  Chern number ($C=-1$) is identical to at least one of the
  post-quench Chern numbers ($C=[0,1,-1]$), hence not all overlaps
  have nodes originating from Thm.~\ref{thm-2d}. Nevertheless, at
  $k_y = 0$ and $\pi$ where the Hamiltonian is real, its eigenstates
  can be classified by their Berry phases. One can find numerically
  that at $k_y = 0$, the Berry phase for $\psi^{(2)}$ is $\gamma = 0$,
  whereas that of the post-quench $\phi^{(3)}$ (the one with $C=-1$)
  is $\gamma = \pi$. According to Thm.~\ref{thm-1d}, therefore,
  $\langle \phi^{(3)} | \psi^{(2)}\rangle_{k_y = 0}(k_x)$ has node
  along $k_x$. Nodes in overlaps of $\psi^{(2)}$ with $\phi^{(1)}$ and
  $\phi^{(2)}$ are still protected by Thm.~\ref{thm-2d}. Thus all
  three overlaps have nodes and DQPT is protected.
\end{inparaenum}

Details of all $18$ quench types ($3$ pre-quench states $\times$ $6$
post-quench phases) can be found in SM. We should note here that out
of all 18 types, 2 robust DQPTs ($\psi^{(1)}$ to phases 2 and 5)
exhibit an even number of overlap nodes at $k_y = 0$ and/or $\pi$ not
accounted for by Thms.~\ref{thm-2d} and \ref{thm-1d}. By tuning
$t_{i,f}$ and $m_{i,f}$, we were able to shift the nodes along $k_x$
as well as to change the total number of nodes by an even number, but
could not entirely eliminate them. We suspect however that they could
eventually be eliminated in an enlarged parameter space.


\begin{figure}
  \centering
  \includegraphics[width=.49\textwidth]{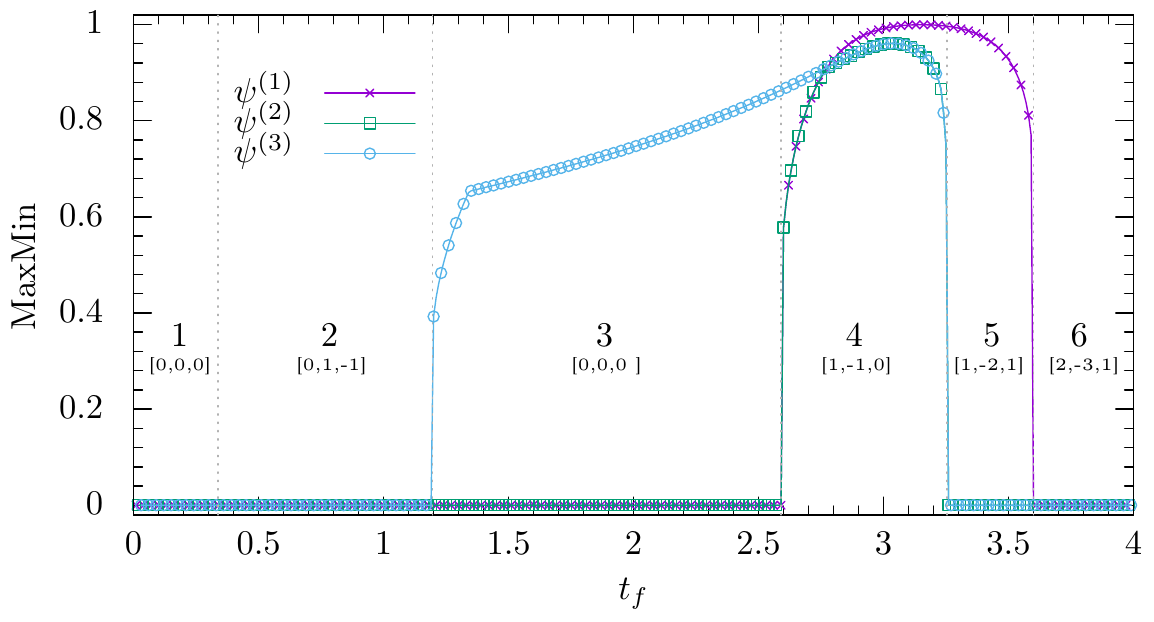}
  \caption{(Color online) Plot of $\psi_{\MaxMin}$ (Eq.~\ref{maxmin1})
    as functions of the post-quench $t$. Pre-quench state is prepared
    by filling one of the three bands $\psi^{(1,2,3)}$ of the
    generalized Hofstadter model Eq.~\ref{hof-hk} with parameters
    $t_i=3$ and $m_i=2.8$. Post-quench $H(\vec k)$ has fixed $m_f = 3$
    and a varying $t_f$, sweeping it through six topological phases
    labeled by its three Chern numbers (ordered from lower to higher
    band). The pre-quench Hamiltonian is in phase 4. A robust DQPT can
    be identified by $\psi_{\MaxMin} = 0$ (Eq.~\ref{maxmin2}). Note
    that $\psi_{\MaxMin}$ changes between zero and non-zero only at
    phase boundaries, verifying robust DQPT as a feature of
    topological phases insensitive to parameter tuning. See SM for
    detailed account of all $18$ types of quenches shown here.}
  \label{fig:maxmin}
\end{figure}

\paragraph{Conclusion and discussion}
In this work, we showed that for quantum quenches between gapped
phases in a generic multi-band system, a robust dynamical quantum
phase transition (DQPT) is a consequence of momentum-space nodes (or
zeros) in the wavefunction overlap between the pre-quench state and
all post-quench energy eigenstates. Nodes in wavefunction overlaps are
topologically protected if the topological indices of the two
participating wavefunctions---such as Chern number in 2D and Berry
phase in 1D---are different.

Our main tenets here are the triangle inequality Eq.~\ref{tri-ineq-k},
and phase ergodicity. It is interesting to note that collapsing a band
gap would affect both conditions: right at the gap collapsing point
$\varepsilon_{\vec k}^{(n)} = \varepsilon_{\vec k}^{(n+1)}$, the two
phases become mutually locked; as the gap re-opens, the system has
gone through a topological transition, which changes the node
structure in wavefunction overlaps. We also note that while the
\emph{existence} of topological and symmetry-protected DQPT is
insensitive to details of the energy band structure, the exact times
at which it would occur will inevitably depend on the latter. The
shortest critical time will be upper-bounded by the recurrence time of
the phases, which, for few-level systems such as band insulators,
should remain physically relevant \footnote{Under special
  circumstances, the DQPT critical time may also emerge at a much
  shorter time scale, e.g., the inverse level spacing, see
  Ref.~\cite{Zhang16}.}.

DQPT in band systems is in principle experimentally measurable. As
shown in Eq.~\ref{rho-kn}, DQPT can be identified as the depletion of
``sublattice'' particle density $\rho_{\vec k,a}(t^{*})$ where
sublattice $a$ can also refer to spin, orbital, etc. Particle density
$\rho_{\vec k}(t) = \sum_a\rho_{\vec k,a}(t)$ can already be measured
in cold atom systems by time-of-flight experiments
\cite{Bloch08,Bloch12,Langen15,Zhao11,Deng14}. It is not hard to
envisage an additional procedure of ``sublattice'' isolation in such
measurements, e.g., by using a magnetic field for spin filtering, or
by releasing other sublattices $b \neq a$ slightly earlier than $a$.

\begin{acknowledgments}
  \paragraph{Acknowledgments} We are grateful to D.P.~Arovas,
  A.~Alexandradinata, and A.~Saxena for useful discussions. We thank
  J.M.~Zhang for communication of a recent work \cite{Zhang16}. Work
  at LANL was supported by the US DOE BES E304/E3B7. Work at NORDITA
  was supported by ERC DM 321031.

  \paragraph{Note added} Ref.~\cite{Gu16} appeared shortly after the
  completion of this manuscript, where results regarding topological
  nodes in wavefunction overlaps, consistent with Theorems
  \ref{thm-2d} and \ref{thm-1d} presented here, are obtained elegantly
  by appealing to adiabatic continuity. In Ref.~\cite{Roy16},
  topological nodes have also been connected to non-analyticity in
  physical observables upon tuning the post-quench Hamiltonian across
  a topological transition. We thank K.~Sun and A.~Das for
  communications.

\end{acknowledgments}

\bibliography{dpt} \bibliographystyle{apsrev-no-url}

\onecolumngrid
\section{Supplemental Materials:}
\setcounter{equation}{0}
\setcounter{theorem}{0}
\setcounter{figure}{0}
\section{Relation between nodes in wavefunction overlaps and
  topological indices}
In this section, we prove two theorems relating momentum-space nodes
in wavefunction overlaps to the topological indices of the two
participating wavefunctions.  \newtheorem{proof}{Proof}
\begin{theorem}
  \label{sm:thm-2d}
  In 2D, the overlap of two Bloch bands $|\psi(\vec k)\rangle$ and
  $|\phi(\vec k)\rangle$, with Chern numbers $C_{\psi}$ and $C_{\phi}$
  respectively, must have at least $|C_{\psi} - C_{\phi}|$ nodes in
  the Brillouin zone. If it has no node, $C_{\psi} = C_{\phi}$.
\end{theorem}
\begin{proof}
  \normalfont Let us first prove the second part of the theorem, where
  $\langle \phi | \psi\rangle$ has no node. Here
  $|\psi(\vec k)\rangle = (\psi_1(\vec k), \psi_2(\vec k), \cdots)^T$
  and
  $|\phi(\vec k)\rangle = (\phi_1(\vec k), \phi_2(\vec k), \cdots)^T$
  are column vectors. Assume $\phi_1$, the first component of $\phi$,
  has $N_{\phi}$ zeros in the Brillouin zone (BZ). Following
  Refs.~\onlinecite{Kohmoto85} and \onlinecite{Hatsugai93}, we
  partition the BZ into $N_{\phi}+1$ regions
  $R_0, R_1, \cdots, R_{N_{\phi}}$, where $R_{i> 0}$ is an
  infinitesimal neighborhood of the $i^{th}$ zero of $\phi_1$, and
  $R_0 = BZ\setminus (\cup_{i=1}^N R_i)$, note that $R_0$ should be
  understood as containing the boundaries $\partial R_{i\neq 0}$. The
  gauge of $|\phi\rangle$ is such that $\phi_1(\vec k) > 0$ for
  $\vec k \in R_0$, denoted as $|\phi\rangle_{R_0}$, and
  $\phi_2(\vec k)>0$ for $\vec k \in R_{i> 0}$ (or any other smooth
  gauge), denoted as $|\phi\rangle_{R_i}$. At the boundaries between
  $R_0$ and $R_{i> 0}$, one has the constituent condition
  \begin{gather}
    |\phi(\vec k)\rangle_{R_0} = e^{i\gamma_{\vec k}^i}|\phi(\vec
    k)\rangle_{R_i}\quad , \quad \vec k \in  \partial R_i\ ,
  \end{gather}
  where $\gamma_{\vec k}^i$ is the gauge mismatch of $\phi$ on the
  boundary of $R_i$. Since
  $\langle \phi(\vec k) | \psi(\vec k)\rangle \neq  0 \forall \vec k$,
  we can choose the gauge of $\psi$ according to that of $\phi$, such
  that their overlap is always real and positive,
  \begin{gather}
    {}_{R_i}\langle \phi(\vec k) | \psi(\vec k)\rangle_{R_i} > 0\quad
    , \quad \vec k \in R_i\ , \ i = 0, 1, 2, \cdots, N_{\phi}\ .
  \end{gather}
  This endows $\psi$ with the same constituent conditions at
  boundaries of $R_{i> 0}$,
  \begin{gather}
    \label{sm:psi-stitch}
    |\psi(\vec k)\rangle_{R_0} = e^{i\gamma_{\vec k}^i} |\psi(\vec
    k)\rangle_{R_i} \quad , \quad \vec k \in  \partial R_i\ .
  \end{gather}
  The Chern number of $\phi$ is given by the total vorticities of the
  stitching phases $\gamma^i$ \cite{Kohmoto85, Hatsugai93},
  \begin{gather}
    \label{sm:vi}
    C_{\phi} = \sum_{i = 1}^{N_{\phi}} v_i\quad , \quad
    v_i\equiv\oint\limits_{\partial R_i} \frac{d\vec k}{2\pi} \cdot
    \nabla_{\vec k} \gamma_{\vec k}^i\ .
  \end{gather}
  Since $\psi$ and $\phi$ share the same set of $\gamma_{\vec k}^i$,
  the two Chern numbers must be identical,
  \begin{gather}
    C_{\psi} = C_{\phi}\ .
  \end{gather}

  The proof of the first part is very similar. If
  $\langle \phi | \psi\rangle$ has $N_{\langle \phi|\psi\rangle}$
  zeros, we instead partition the BZ into
  $N_{\phi} + N_{\langle\phi|\psi\rangle} + 1$ regions, where
  $R_1, \cdots, R_{N_{\phi}}$ are defined the same as before, and one
  has in addition the regions $R_{N_{\phi} + j}$,
  $j = 1, 2, \cdots, N_{\langle \phi | \psi\rangle}$ surrounding nodes
  of $\langle \phi | \psi\rangle$.
  $R_0 = BZ \setminus
  (\cup_{i=1}^{N_{\phi}+N_{\langle\phi|\psi\rangle}} R_i)$.  The gauge
  choice of $\phi$ remains the same as before. For $\psi$, one can
  still choose its gauge as $\langle \phi | \psi\rangle > 0$ in
  $R_0, R_1, \cdots, R_{N_{\phi}}$. For the additional regions
  $R_{N_{\phi} + n}$, $n=1,2,\cdots, N_{\langle \phi | \psi\rangle}$,
  one can choose any other smooth gauge, say, the component
  $\psi_1 > 0$, which induce $N_{\langle \phi | \psi\rangle}$ new
  stitching phases. The Chern number of $\psi$ is the total vorticity
  of the $N_{\phi} + N_{\langle \phi|\psi\rangle}$ phases,
  \begin{gather}
    C_{\psi} = \sum_{i=1}^{N_{\phi} + N_{\langle \phi | \psi\rangle}}
    v_i = C_{\phi} + \sum_{i=N_{\phi}+1}^{N_{\phi} + N_{\langle \phi |
        \psi\rangle}} v_i\ ,
  \end{gather}
  where Eq.~\ref{sm:vi} is used to get the second equality. This implies
  that
  \begin{gather}
    \label{sm:c-diff}
    \left| \sum_{i=N_{\phi}+1}^{N_{\phi} + N_{\langle \phi |
          \psi\rangle}} v_i\right| = \left| C_{\psi} - C_{\phi} \right|\ .
  \end{gather}
  Note that the LHS of Eq.~\ref{sm:c-diff} satisfies
  $|\sum v_i | \le \sum |v_i| \le N_{\langle \phi | \psi\rangle}$,
  where we used $|v_i| = 0, \pm 1$, taking the viewpoint that vortices
  of vorticity $|v| > 1$ can be thought of as consisting of
  overlapping vortices of vorticity $\pm 1$. Thus
  \begin{gather}
    N_{\langle \phi | \psi\rangle} \ge |C_{\psi} - C_{\phi}|\ .
  \end{gather}
\end{proof}
\begin{theorem}
  \label{sm:thm-1d} 
  In 1D, the Berry phase $\gamma$ of a real Bloch band,
  $|\psi(k)\rangle = (\psi_1(k), \psi_2(k), \cdots)^t, \psi_a(k) \in
  \mathbb{R} \,\forall a$, is quantized to $0$ or $\pi$. The overlap
  of two real bands $|\psi(k)\rangle$ and $|\phi(k)\rangle$, with
  Berry phases $\gamma_{\psi}$ and $\gamma_{\phi}$ respectively, must
  have at least one node if $\gamma_{\psi} \neq \gamma_{\phi}$.
\end{theorem}
\begin{proof}
  \normalfont Let us prove the first part first. The Berry phase of a
  band $|\psi(k)\rangle$ is defined as
  \begin{gather}
    \gamma = \int\limits_0^{2\pi} dk \langle \psi(k) | i \partial_k
    \psi(k)\rangle\  \mod\ 2\pi.
  \end{gather}
  Although this definition is not restricted to real bands, in our
  case, $|\psi(k)\rangle$ is a real band.

  Assume the first component of $\psi$, $\psi_1(k)$, has $N_{\psi}$
  zeros at $k_1, k_2, \cdots, k_{N_{\psi}}$. We fix the gauge of
  $|\psi(k)\rangle$ such that $\psi_1(k) > 0$ for
  $k \notin \{k_1, k_2, \cdots, k_{N_{\psi}}\}$, and $\psi_2(k) > 0$
  for $k \in \{k_1, k_2, \cdots, k_{N_{\psi}}\}$. Under this gauge,
  the Berry phase accumulated within each nodeless segment
  $k: k_i + 0^+ \rightarrow k_{i+1} - 0^+$ vanishes. On the two sides
  of a nodal point $k_i$, however, the state $|\psi\rangle$ may differ
  by an overall sign, giving rise to a $\pi$ Berry phase accretion
  across $k_i$. The total Berry phase $\gamma_{\psi}$ is thus
  \begin{gather}
    \label{sm:berry-phi}
    e^{i\gamma_{\psi}} = \prod_{i=1}^{N_{\psi}} v_i\quad , \quad v_i =
    \langle \psi(k_i - 0^+) | \psi(k_i + 0^+)\rangle = \pm 1\ .
  \end{gather}
  This is the analogue of Eq.~\ref{sm:vi} of the 2D case. Thus the Berry
  phase $\gamma_{\psi}$ is quantized to either $0$ or $\pi$.

  We note that the quantized Berry phase is equivalent to the $Z_2$
  index introduced in Ref.~\onlinecite{Fang15}.

  Proof of the second part is very similar to the 2D case, where
  imposing $\langle \phi | \psi\rangle > 0$ induces a gauge for $\phi$
  based on that of $\psi$. If the overlap of two real bands
  $\langle\phi(k) | \psi(k)\rangle$ has no nodes, then imposing
  $\langle \phi(k) | \psi(k)\rangle > 0$ everywhere enforces the sign
  changes of $|\phi(k)\rangle$ to be synchronized with those of
  $|\psi(k)\rangle$ at $\{k_1, k_2, \cdots, k_{N_{\phi}}\}$, so they
  must have the same Berry phase. Thus if their Berry phases are
  different, there must exist node in their overlap.

\end{proof}

\section{A 1D 3-band model with symmetry-protected DQPT}
To illustrate Thm.\ref{sm:thm-1d}, consider the following three band
model in 1D,
\begin{gather}
  \label{ssh3-h}
  H(k,t,m) =
  \begin{pmatrix}
    \sin k + m & \cos k + t \\
    \cos k + t & -\sin k & \cos k + t\\
    & \cos k + t & \sin k - m
  \end{pmatrix}\ .
\end{gather}
This model describes three coupled chains arranged as an $N\times 3$
square lattice, where $N$ is the length of the chains along the $x$
direction (assumed periodic), with horizontal (intra-chain nearest
neighbor) and diagonal (next nearest neighbor) hopping $1$, vertical
(inter-chain nearest neighbor) hopping $t\ge 0$, on-site mass
modulation $m\ge 0$, and a $\pi$ flux per square plaquette. Note that
for $m = 0$ and when projected onto the subspace of the first two
chains, $P=|1\rangle\langle 1 |+ |2\rangle\langle 2|$, the Hamiltonian
reduces to $PH(k)P = \sin k \, \sigma_z + (\cos k + t) \sigma_x$,
which, under the rotation $\sigma_z \rightarrow \sigma_y$, becomes the
celebrated Su-Schrieffer-Heeger model. The central band always has
Berry phase $0$. For infinitesimal $m$, the topological transition
occurs at $t \simeq 1$.

Since $H(k,t,m)$ is real, all three bands are subject to
Thm.~\ref{sm:thm-1d}. For a fixed $m$, the top and bottom bands have
Berry phase $\pi$ if $H(k,t,m)$ is adiabatically connected with
$H(k, t\rightarrow 0, m)$, or Berry phase $0$ if $H(k,t,m)$ is
adiabatically connected with $H(k, t\rightarrow \infty, m)$. 

\subsection{Berry phase for $m \rightarrow 0^+$}
While the Hamiltonian Eq.~\ref{ssh3-h} for generic $t$ and $m$ does
not lend itself easily to closed-form diagonalization, the case with
$m \rightarrow 0^{+}$ can be analytically solved. Topological
classification of generic $t$ and $m$ can then be obtained, in
principle, by adiabatic continuation (i.e. keeping both gaps open).

When $m = 0$, band degeneracy occurs at $k = \pm k_c$ if $t < 1$, with
$k_c \ge 0$ determined by
\begin{gather}
  \cos k_c + t = 0\ .
\end{gather}
For $m \rightarrow 0^+$, one instead has two avoided crossings at $\pm k_c$.

Away from the avoided crossings, the eigenstates are given by
those of $H(k, t, m=0)$,
\begin{gather}
  |\psi^{\pm}_k\rangle =
  \begin{pmatrix}
    \cos k + t \\
    -\sin k \pm \sqrt{2(\cos k + t)^2 + \sin^2 k} \\
    \cos k + t
  \end{pmatrix}\quad , \quad |\psi^3_k\rangle =
  \begin{pmatrix}
    1 \\ 0 \\ -1
  \end{pmatrix}\ , \\
  E_k^{\pm} = \pm \sqrt{2(\cos k + t)^2 + \sin^2 k}\quad , \quad E_k^3 = \sin k\ .
\end{gather}
Note that the eigenstates are \emph{un-normalized}. $E^{\pm}$ and
$E^3$ are the corresponding eigenvalues. Close to, but not at, the
avoided crossings, one can write $k = \pm k_c + \delta$ and expand
$\psi^{\pm}$ to first order of $\delta$,
\begin{gather}
  \label{psi-pm-del}
  |\psi^+_{k_c + \delta}\rangle = -\sin k_c\ \delta
  \begin{pmatrix}
    1 \\ 0 \\ 1
  \end{pmatrix} + \mathcal{O}(\delta^2)
  \quad,\quad 
  |\psi^-_{k_c + \delta }\rangle = -2\sin k_c
  \begin{pmatrix}
    0 \\ 1 \\ 0
  \end{pmatrix}
  - \delta 
  \begin{pmatrix}
    \sin k_c \\ 2 \cos k_c \\ \sin k_c
  \end{pmatrix} 
+ \mathcal{O}(\delta^2) \ ,\\
  |\psi^+_{-k_c + \delta}\rangle = 2\sin k_c
  \begin{pmatrix}
    0 \\ 1 \\ 0
  \end{pmatrix} + \delta
  \begin{pmatrix}
    \sin k_c \\ -2\cos k_c \\ \sin k_c
  \end{pmatrix}
+ \mathcal{O}(\delta^2)
  \quad,\quad 
  |\psi^-_{-k_c + \delta }\rangle = \sin k_c \ \delta
  \begin{pmatrix}
    1 \\ 0 \\ 1
  \end{pmatrix} + \mathcal{O}(\delta^2) \ .
\end{gather}

While the un-normalized $\psi^{\pm}$ are smooth over the entire BZ,
their normalized versions would experience a phase (i.e. gauge)
mismatch near the avoided crossings, which gives rise to Berry
phase. Consider, for example, the Berry phase of $\psi^+$. From the
first equation in Eq.~\ref{psi-pm-del}, the normalized $\psi^+$ is
$|\wt \psi^+_{k_c + \delta}\rangle = -\Sgn(\delta) \smat{1 \\ 0 \\
  1}/\sqrt{2}$. It picks up an overall negative sign going from
$k_c - 0^+$ to $k_c + 0^+$, i.e. a $\pi$ Berry phase change. Anywhere
else its gauge is smooth, therefore $\psi^+$ has a $\pi$ Berry phase.

Similar analysis shows that $\psi^-$ acquires a $\pi$ Berry phase
passing through $-k_c$.

For $t > 1$, $\cos k + t$ never reaches zero, hence $\psi^{\pm}$ both
have a smooth and periodic gauge and therefore their Berry phases are
zero.

The Berry phase of the central band is always zero for any $t$ as
$|\psi^3_k\rangle$ is $k$-independent. This is consistent with the
fact that the total Berry phase of all bands (for any band models) is
always zero.

\subsection{DQPT in quench from topological to trivial phase}
We prepare the pre-quench state by filling the lowest band in the
topological phase $H_I(k) = H(k, t \rightarrow 0, m)$,
\begin{gather}
  |\Psi\rangle = \prod_k (\psi^{-}_k)\dg | \emptyset\rangle\ .
\end{gather}
At time $t=0$, we switch the Hamiltonian to the following
topologically trivial one,
\begin{gather}
  \label{ssh3-trivial}
  H_F(k) = H(k,t= t_F\rightarrow\infty, m) = t_F
  \begin{pmatrix}
    0 & 1 & 0 \\ 1 & 0 & 1 \\ 0 & 1 & 0
  \end{pmatrix}\ .
\end{gather}
Its eigenstates and eigenvalues are
\begin{gather}
  |\phi^{\pm}_k\rangle = \frac{1}{2}
  \begin{pmatrix}
    1 \\ \pm \sqrt{2} \\ 1
  \end{pmatrix}\quad , \quad |\phi^3_k\rangle = \frac{1}{\sqrt{2}}
  \begin{pmatrix}
    1 \\ 0 \\ -1
  \end{pmatrix}\quad , \quad E^{\pm}_k = \pm \sqrt{2}\, t_F\quad , \quad E^3_k = 0\ .
\end{gather}
Fig.~\ref{fig:ssh3-overlap} plots the overlap of the pre-quench band
with all three post-quench eigenstates, verifying that all three
overlaps have node as required by Thm.~\ref{sm:thm-1d}. These nodes in
turn guarantees the DQPT condition
$ \exists k: |\langle \phi^{(n)}(k)|\psi(k)\rangle|^2 \le \frac{1}{2}\
\forall n$.
\begin{figure}
  \centering
  \includegraphics[width=.5\textwidth]{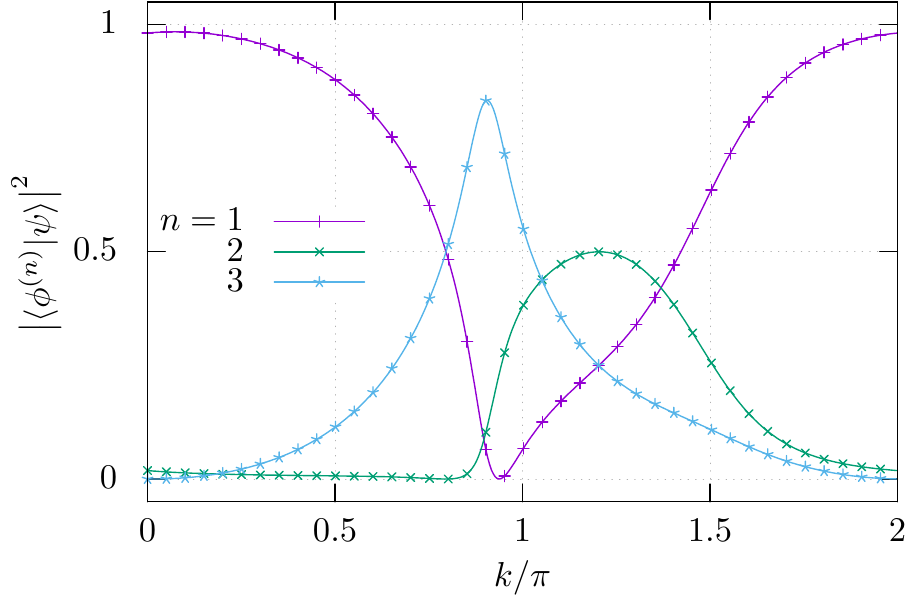}
  \caption{Overlap between pre-quench band and post-quench eigenstates
    for the 3-band model Eq.~\ref{ssh3-h}. Pre-quench parameters are
    $t_i = 0.8, m_i = 0.5$. Its lowest band has Berry phase $\pi$ and
    is used as the pre-quench state. Post-quench Hamiltonian is the
    $t_f \rightarrow \infty$ limit (cf.~Eq.~\ref{ssh3-trivial}), all
    eigenstates $\phi^{(1,2,3)}$ there have $0$ Berry phase. According
    to Thm.~\ref{sm:thm-1d}, all three overlaps must have nodes, which
    can be verified from the plot. The DQPT condition
    $ \exists k: |\langle \phi^{(n)}(k)|\psi(k)\rangle|^2 \le
    \frac{1}{2}\ \forall n$ is thus protected.}
  \label{fig:ssh3-overlap}
\end{figure}
\section{Generalized Hofstadter model}
A generalization of the $q$-band Hofstadter model is
\begin{gather}
  H(k_x,k_y,p,q,t,m) =
  \begin{pmatrix}
    d_1 & v_1 & & & & v_qe^{ik_y} \\
    v_1 & d_2 & v_2 \\
    & v_2 &\ddots &\ddots\\
    & & \ddots & \ddots\\
    & & & & & v_{q-1}\\
    v_qe^{-ik_y} & & & & v_{q-1} & d_q
  \end{pmatrix}\ , \\
  d_a = 2\cos(k_x + a\phi) + a\,m\quad , \quad v_a = 1 +
  2t\cos\left[k_x + (a+\frac{1}{2})\phi\right]\quad , \quad a = 1, 2,
  \cdots, q\quad , \quad \phi = 2\pi \frac{p}{q}\ .
\end{gather}
Here $p$ and $q$ are co-prime integers, $\phi$ is the flux per square
plaquette, the magnetic unit cell is chosen along the $y$ direction,
$t$ is the amplitude of the second neighbor hopping (with nearest
neighbor hopping set to $1$). The Pierls phases along the diagonals
(i.e. of the second neighbor hops) are chosen so that the flux per
triangle is $\phi/2$. Note that for $q = 2$, the off-diagonal matrix
element is $v_1 + v_q e^{ik_y}$. We also impose a sawtooth potential
in the $y$-direction, $V_y = (y \mod q)\times m$, which is
commensurate with the magnetic unit cell. 

\begin{figure}
  \centering
  \includegraphics[width=0.7\textwidth]{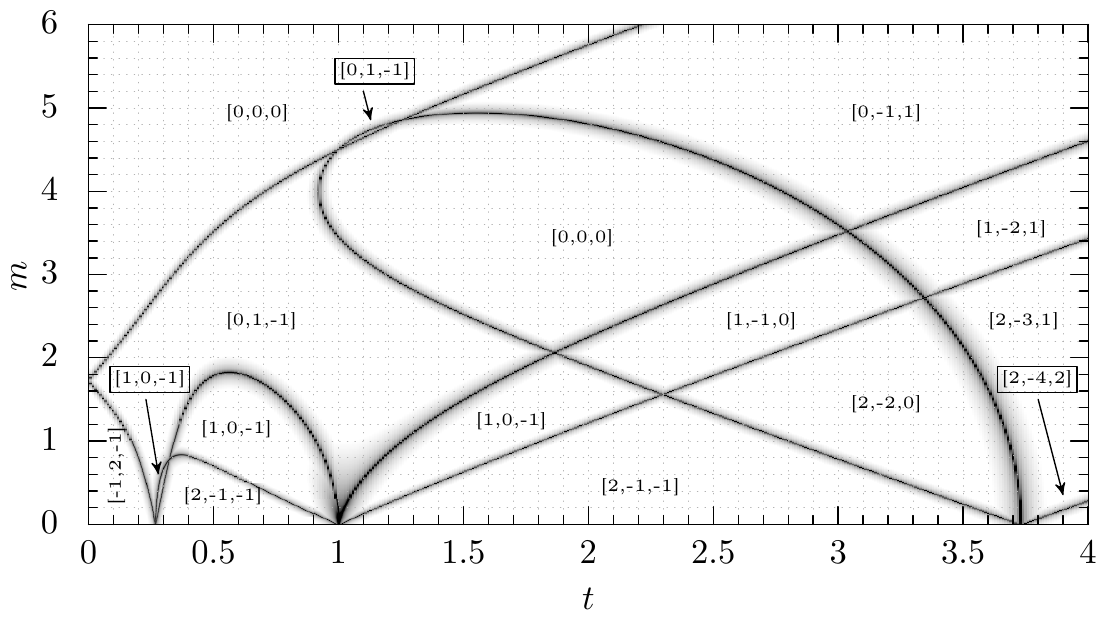}
  \caption{Topological phase diagram of the $3$-band Hofstadter model
    in the $t-m$ space. Black lines are phase boundaries (where at
    least one of the two spectral gaps collapses). Unboxed labels
    denote the Chern numbers (from lower to higher energy bands) of
    the corresponding phases. Boxed labels denote the Chern numbers in
    the phase indicated by the arrows. Only $t\ge 0$ and $m \ge 0$ are
    shown, other quadrants can be obtained by its mirror image, which
    follows from (1) that the Chern numbers remain the same for
    $m \rightarrow -m$, and (2) that the Chern numbers reverse order
    for $t \rightarrow -t$. }
  \label{sm:fig:t-m-phase-diagram}
\end{figure}

\subsection{Phase diagram}
The $3$-band model used in the text amounts to fixing $p=1,
q=3$. Fig.~\ref{sm:fig:t-m-phase-diagram} shows its topological phase
diagram in the $t-m$ plane. We only plot the $t\ge 0$ and $m \ge 0$
quadrant, the other three quadrants can be obtained by taking its
mirror image, noting that (1) the Chern numbers remain the same for
$m \rightarrow -m$, and that (2) the Chern numbers reverse order for
$t \rightarrow -t$.

\subsection{Protected vs unprotected DQPT}
\begin{figure}
  \centering \subfloat[$t_f = 1.4$, No DQPT. \newline Same post-quench phase as
  (b)]{
    \includegraphics[width=.35\textwidth]{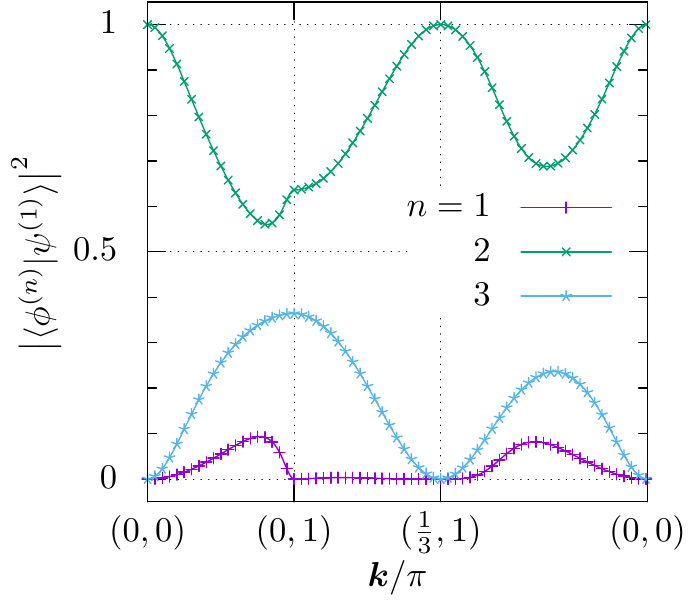}
  }
  \centering
  \subfloat[$t_f = 3$, Unprotected DQPT.\newline Same post-quench phase as (a)]{
    \includegraphics[width=.316\textwidth,trim={20 0 0 0},clip]{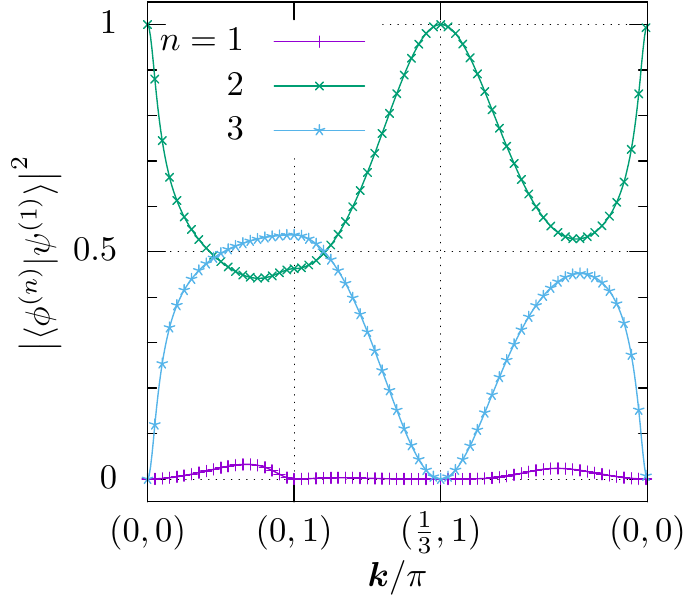}
  }
  \centering
  \subfloat[$t_f = 4$, Protected DQPT. \newline Different post-quench phase from (a) and (b)]{
    \includegraphics[width=.316\textwidth,trim={20 0 0 0},clip]{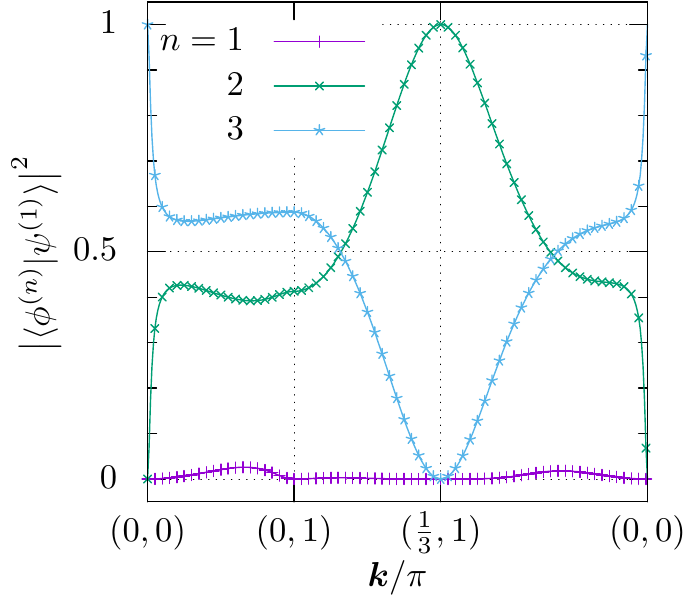}
  }
  \caption{Protected vs unprotected DQPT of the $3$-band generalized
    Hofstadter model. In all three panels, we fix the pre-quench
    Hamiltonian $H_I$ with $t_i = -0.8, m_i = 0$, and use three
    different $t_f$, as shown in panel captions, for the post-quench
    Hamiltonian, while keeping $m_f = 0$. $H_I$ is in the topological
    phase with Chern numbers $[-1,-1,2]$ for its three bands. For
    panels (a) and (b), $H_F$ is in the phase $[2,-1,-1]$, and for
    panel (c), it is in phase $[2,-4,2]$. We plot the overlap of
    $\psi^{(1)}$ with $\phi^{(1,2,3)}$ along high symmetry lines (on
    which overlap minima will occur because we have set
    $m_i = m_f = 0$). Panel (b) illustrates unprotected crossing of
    $n=2$ and $3$ near $\frac{1}{2}$, at which point the DQPT
    amplitude condition
    $\bigl|\langle \phi^{(n)}(\vec k) | \psi^{(1)}(\vec
    k)\rangle\bigr|^2 \le \frac{1}{2} \forall n$ is satisfied, yet
    since the $n=2$ line (green line with cross) never reaches zero,
    this crossing is not robust, and can be removed by tuning $H_I$
    and $H_F$ within their respective phases. This type of avoided
    crossing is shown in (a), note that in this case the $n=2$ line
    lies entirely above $\frac{1}{2}$, hence DQPT is forbidden. In
    (c), all overlaps touch zero, thus the crossing near $\frac{1}{2}$
    and DQPT is robust.}
  \label{sm:fig:overlap-kline}
\end{figure}

While the DQPT amplitude condition only require the existence of at
least one $\vec k$ point at which all overlaps
$\bigl|\langle \phi^{(n)}(\vec k) | \psi(\vec k)\rangle\bigr|^2 \le
\frac{1}{2}\forall n$, it is not a phase-robust feature. This is
illustrated in Panels (a) and (b) in Fig.~\ref{sm:fig:overlap-kline},
where both are in the same pre- and post-quench topological phases,
yet simply by tuning the post-quench parameter $t_f$ the DQPT in (b)
can be avoided. Therefore, a protected DQPT requires the existence of
$\vec k$-space nodes in all overlaps, which would guarantee the
amplitude condition. This is demonstrated in
Fig.~\ref{sm:fig:overlap-kline} (c).

\subsection{Detailed analysis of Fig.~(2) in text}
Here we analyze all quench types shown in Fig.~2 in the main text,
reproduced here in Fig.~\ref{sm:fig:maxmin} for convenience.

\begin{figure}
  \centering
  \includegraphics[width=.7\textwidth]{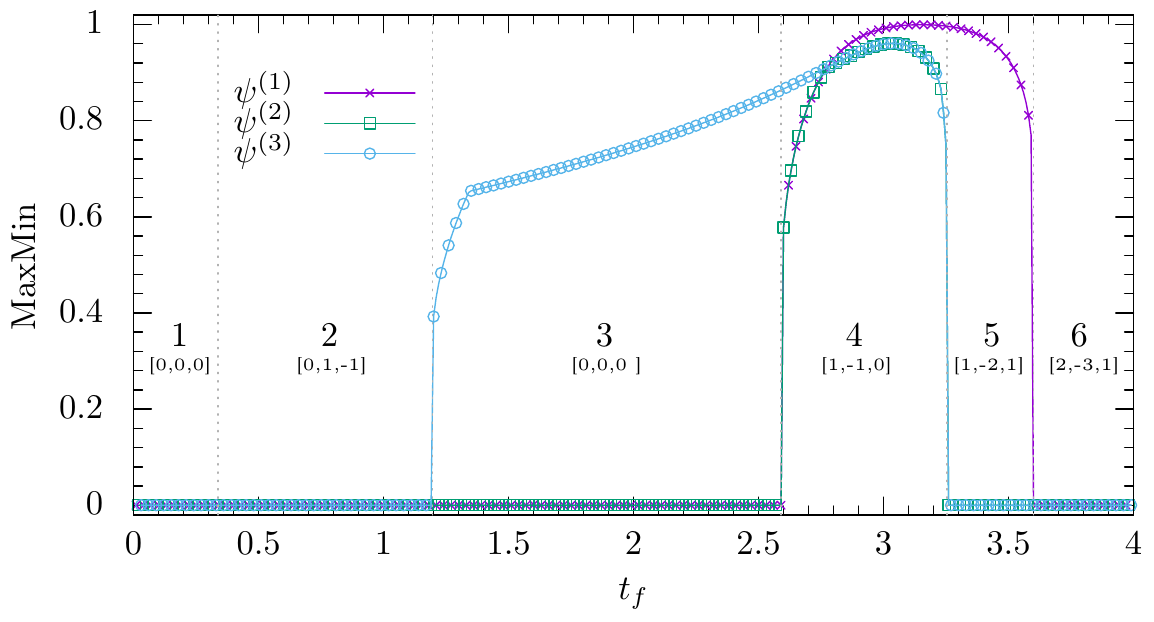}
  \caption{(Color online) Plot of $\psi_{\MaxMin}$ (see main text for
    definition) as functions of the post-quench $t$. Pre-quench state
    is prepared by filling one of the three bands $\psi^{(1,2,3)}$ of
    the $3$-band generalized Hofstadter model with parameters $t_i=3$
    and $m_i=2.8$. Post-quench $H(\vec k)$ has fixed $m_f = 3$ and a
    varying $t_f$, sweeping it through six topological phases labeled
    by its three Chern numbers (ordered from lower to higher
    band). The pre-quench Hamiltonian is in phase 4. A robust DQPT can
    be identified by $\psi_{\MaxMin} = 0$ (see main text). Note that
    $\psi_{\MaxMin}$ changes between zero and non-zero only at phase
    boundaries, verifying robust DQPT as a feature of topological
    phases insensitive to parameter tuning. }
  \label{sm:fig:maxmin}
\end{figure}

There are 18 quench types shown in Fig.~\ref{sm:fig:maxmin} (3 pre-quench
bands, 6 post-quench phases). Their relation with DQPT is summerized
in the table below,
\begin{gather}
  \label{sm:maxmin-table}
  \begin{tabular}{l|c c c c c c}
    & 1 & 2 & 3 & 4 & 5 & 6\\
    & [0,0,0] & [0,1,-1] & [0,0,0] & [1,-1,0] & [1,-2,1] & [2,-3,1]\\
    \hline
    $\psi^{(1)}$ \normalfont{($C=1$) }& Y & A & Y & X & N & A\\
    $\psi^{(2)}$ \normalfont{($C=-1$)}& Y & S & Y & X & Y & Y\\
    $\psi^{(3)}$ \normalfont{($C=0$) }& S & S & N & X & Y & Y
  \end{tabular}\ ,
\end{gather}
where we have
\begin{itemize}
\item \emph{Topological DQPT}: $8$ entries with ``Y''. As discussed in
  the main text, these are quenches with DQPT due to all post-quench
  Chern numbers being different from the Chern number of the
  pre-quench band. Thm.~\ref{sm:thm-2d} then guarantees the overlap nodes.
\item \emph{Symmetry protected DQPT}: $3$ entries with ``S''. As
  discussed in the main text, these are quenches where at least one of
  the post-quench Chern number is identical to that of the pre-quench
  band, yet on high symmetry lines $k_y = 0$ and $\pi$, their overlap
  still exhibits nodes protected by Berry phase. This subclass will be
  discussed in detail in the next section.
\item $3$ entries with ``X'', which denote quenches with no robust
  DQPT because the pre-quench and post-quench Hamiltonians are in the
  same topological phase.
\item $2$ entries with ``N'', which denote quenches with no robust
  DQPT, because the topological transition between the pre-quench and
  post-quench Hamiltonians only involves collapsing a gap which is not
  immediate to the pre-quench band (i.e., the pre-quench band never
  touches other bands over the topological transition). For example,
  when quenching the lowest band $\psi^{(1)}$ toward phase $5$ (with
  final states $\phi^{1,2,3}$), $\psi^{(1)}$ and $\phi^{(1)}$ are
  adiabatically related (i.e. without collapsing the lowest spectral
  gap in $H(\vec k)$), hence it is understandable that there is no
  node in their overlap, which precludes a robust DQPT.
\item $2$ entries with ``A''. As discussed in the main text, these
  quenches exhibit an even number of overlap nodes at $k_y = 0$ and/or
  $\pi$, but these nodes cannot be accounted for by Thms.~\ref{sm:thm-2d}
  and \ref{sm:thm-1d} (see Eq.~\ref{sm:berry} for Berry phases). By
  tuning $t_{i,f}$ and $m_{i,f}$, we were able to shift the nodes
  along $k_x$ as well as to change the total number of nodes by an
  even number, but could not entirely eliminate them. We suspect
  however that they could eventually be eliminated in an enlarged
  parameter space
\end{itemize}

\subsection{Symmetry-protected DQPT}
\label{sm:sec:spdpt-hof}
There are three entries in Eq.~\ref{sm:maxmin-table} (see also
Fig.~\ref{sm:fig:maxmin}) where the occurrence of DQPT is due to
symmetry protection. For example, consider the quench from the highest
band $|\psi^{(3)}\rangle$ of phase $4$ toward phase $1$. The Chern
number of the pre-quench band and all three post-quench bands are $0$,
thus there is no node protected by Thm.~\ref{sm:thm-2d}. However,
since $H(k_x, -k_y) = H(k_x, k_y)^{*}$, the eigenstates of $H$ at
$k_y = 0$ and $\pi$ are real, and can hence be classified by their
respective Berry phases over the $k_x$ loop. Numerically one finds the
Berry phases to be
\begin{gather}
  \label{sm:berry}
  \begin{array}{c||c|c|c|c|c|c}
    & 1 & 2 & 3 & 4 & 5 & 6\\
    \hline
    \psi^{(1)} & 0,0 & 0,0 & 0,0 & 0,\pi & 0,\pi & \pi,\pi\\
    \psi^{(2)} & 0,0 & 0,\pi & \pi,\pi & \pi,0 & 0,0 & \pi,0 \\
    \psi^{(3)} & 0,0 & 0,\pi & \pi,\pi & \pi,\pi & 0,\pi & 0,\pi
  \end{array}\ ,
\end{gather}
where the ordered number pair in each cell denotes the Berry phases at
$k_y = 0$ and $\pi$, respectively, of a state (row) in a topological
phase (column). In Fig.~\ref{sm:fig:maxmin}, the pre-quench
Hamiltonian is in phase $4$, where the two Berry phases of
$\psi^{(3)}$ are both $\pi$, whereas those of the post-quench states
(phase $1$) are all zero. Thus while $\psi^{(3)}$ as well as
$\phi^{1,2,3}$ are all trivial in the quantum Hall sense, they
nonetheless belong to different classes at the high symmetry $k_y$'s,
hence the nodes in their overlaps are protected by
Thm.~\ref{sm:thm-1d}, leading to a protected DQPT.

\end{document}